\documentclass[11pt]{article}

\usepackage{fullpage,authblk,times,latexsym,amsmath,amsfonts,amssymb}

\def\01{\{0,1\}}
\newcommand{\ceil}[1]{\lceil{#1}\rceil}

\newcommand{\eps}{\varepsilon}

\newcommand{\ket}[1]{|#1\rangle}

\newcommand{\inp}[2]{\langle{#1}|{#2}\rangle} 

\newcommand{\OR}{\mbox{\rm OR}}
\newcommand{\ED}{\mbox{\rm ED}}
\newcommand{\ADV}{\mbox{\rm ADV}}
\newcommand{\kSUM}{\mbox{\rm $k$-sum}}
\newcommand{\LGC}{\mbox{\rm LGC}}
\newcommand{\norm}[1]{\mbox{$\lVert{#1}\rVert$}}

\newtheorem{theorem}{Theorem}
\newtheorem{lemma}[theorem]{Lemma}
\newtheorem{fact}[theorem]{Fact}
\newtheorem{corollary}[theorem]{Corollary}

\newtheorem{definition}{Definition}

\newenvironment{proof}
{\noindent {\bf Proof. }}
{{\hfill $\Box$}\\
 \smallskip}

\begin{document}

\title{Optimal parallel quantum query algorithms\thanks{Partially supported by the French ANR Blanc project ANR-12-BS02-005 (RDAM), a Vidi grant from the Netherlands Organization for Scientific Research (NWO), ERC Consolidator grant QPROGRESS, the European Commission IST STREP projects Quantum Computer Science (QCS) 255961, Quantum Algorithms (QALGO) 600700, and the US ARO. An extended abstract of this paper appeared in the Proceedings of the 22nd European Symposium on Algorithms (ESA'14), pp.592-604.}}

\author[1]{Stacey Jeffery%
\thanks{\texttt{sjeffery@caltech.edu}}}
\author[2]{Frederic Magniez%
\thanks{\texttt{frederic.magniez@univ-paris-diderot.fr}}}
\author[3]{Ronald de Wolf%
\thanks{\texttt{rdewolf@cwi.nl}}}
\affil[1]{David R.\ Cheriton School of Computer Science and Institute for Quantum Computing, University of Waterloo, Canada}
\affil[2]{CNRS, LIAFA, Univ Paris Diderot, Sorbonne Paris-Cit\'e, 75205 Paris, France}
\affil[3]{CWI and University of Amsterdam, the Netherlands}

\date{}
\maketitle
\begin{abstract}
We study the complexity of quantum query algorithms that make $p$ queries in parallel in each timestep.  
This model is in part motivated by the fact that decoherence times of qubits are typically small, so it makes sense to parallelize quantum algorithms as much as possible. 
We show tight bounds for a number of problems, specifically $\Theta((n/p)^{2/3})$ $p$-parallel queries for element distinctness and $\Theta((n/p)^{k/(k+1)})$ for \kSUM. Our upper bounds are obtained by parallelized quantum walk algorithms, and our lower bounds are based on a relatively small modification of the adversary lower bound method, combined with recent results of Belovs et al.\ on learning graphs. We also prove some general bounds, in particular that quantum and classical $p$-parallel complexity are polynomially related for all total functions~$f$ when $p$ is small compared to $f$'s block sensitivity.
\end{abstract}
\thispagestyle{empty} 
\setcounter{page}{0} 
\newpage{}

\section{Introduction}

\subsection{Background}
Using quantum effects to speed up computation has been a prominent research-topic for the past two decades.
Most known quantum algorithms have been developed in the model of quantum query complexity,
 the quantum generalization of decision tree complexity.
Here an algorithm is charged for each ``query'' to the input,
while intermediate computation is free (see~\cite{buhrman&wolf:dectreesurvey} for more details). This model facilitates the proof of lower bounds, and often, though not always, quantum query upper bounds carry over to quantum time complexity. 
For certain functions one can obtain large quantum-speedups in this model.
For example, Grover's algorithm~\cite{grover:search} can search an $n$-bit database (looking for a bit-position of a~1) using $O(\sqrt{n})$ queries, while any classical algorithm needs $\Omega(n)$ queries.
For some partial functions we know exponential and even unbounded speed-ups~\cite{deutsch&jozsa,simon:power,shor:factoring,bcw:sharp,aaronson&ambainis:forrelation}.

A more recent crop of quantum speed-ups come from algorithms based on \emph{quantum walks}. Such algorithms solve a search problem by embedding the search on a graph, and doing a quantum walk on this graph that converges rapidly to a superposition over only the ``marked'' vertices, which are the ones containing a solution. An important example is Ambainis's quantum algorithm for solving the \emph{element distinctness} problem~\cite{ambainis:edj}.  In this problem one is given 
an input $x\in[q]^n$, and the goal is to find a pair of distinct $i$ and~$j$ in $[n]$ such that $x_i=x_j$, or report that none exists. Ambainis's quantum walk solves this in $O(n^{2/3})$ queries, which is optimal~\cite{aaronson&shi:collision}. Classically, $\Theta(n)$ queries are required.
Two generalizations of this are the \emph{$k$-distinctness} problem, where the objective is to find distinct $i_1,\ldots,i_k\in[n]$ such that $x_{i_1}=\cdots=x_{i_k}$, and the \emph{$k$-sum} problem, where the objective is to find distinct $i_1,\ldots,i_k\in[n]$ such that $x_{i_1}+\cdots+x_{i_k}=0\mod q$.  Ambainis's approach solves both problems using $O(n^{k/(k+1)})$ quantum queries.
Recently, Belovs gave a $o(n^{3/4})$-query algorithm for $k$-distinctness for any fixed~$k$~\cite{belovs:kdistinctness} (which can also be made time-efficient for $k=3$~\cite{bcjkm:timeefficient}). In contrast,  
Ambainis's $O(n^{k/(k+1)})$-query algorithm is optimal for $k$-sum~\cite{belovs:ksum,belovs&spalek:ksum}.

Here we consider to what extent such algorithms can be \emph{parallelized}. Doing operations in parallel is a well-known way to trade hardware for time, speeding up computations by distributing the work over many processors that run in parallel.  This is becoming ever more prominent in classical computing due to multi-core processors and grid computing.  In the case of quantum computing there is an additional reason to consider parallelization, namely the limited lifetime of qubits due to \emph{decoherence}: because of unintended interaction with their environment, qubits tend to lose their quantum properties over a limited amount of time, called the \emph{decoherence time}, and degrade to classical random bits. One way to fight this is to apply the recipes of 
quantum error-correction and fault-tolerance%
\footnote{Parallelism is in fact \emph{necessary} to do quantum error-correction against a constant noise rate: sequential operations cannot keep up with the parallel build-up of errors.}, which can counteract the effects of 
certain kinds of decoherence.  Another way is to try to parallelize as much as possible, completing the computation before the qubits decohere too much (this may of course increase the width of the computation, creating problems of its own).
\subsection{Earlier work on parallel quantum algorithms}
We know of only a few results about parallel quantum algorithms, most of them in the circuit model where ``time'' is measured by the depth of the circuit.  A particularly important and beautiful example is the work of Cleve and Watrous~\cite{cleve&watrous:parfourier}, who showed how to implement the $n$-qubit quantum Fourier transform using a quantum circuit of depth $O(\log n)$.  As a consequence, they were able to parallelize the quantum component of Shor's algorithm: they showed that one can factor $n$-bit integers by means of an $O(\log n)$-depth quantum circuit with polynomial-time classical pre- and post-processing. 
There have also been a number of papers about quantum versions of small-depth classical Boolean circuit classes like AC and NC~\cite{moore&nilsson:parallelqc,ghmp:qacc,hoyer&spalek:fanout,takahashi&tani:collapse}. 
Beals et al.~\cite{bbghklss:distributedq} show how the quantum circuit model can be efficiently simulated by the more realistic model of a distributed quantum computer (see also~\cite{grover&rudolph}).
The setting of \emph{measurement-based} quantum computing (see~\cite{jozsa:measbased} and references therein) in some cases allows more parallelization than the usual circuit model. 

Another example, basically the only one we know of in the setting of query complexity, is Zalka's tight analysis of parallelizing quantum search~\cite[Section~4]{zalka:grover}.  Suppose one wants to search an $n$-bit database, with the ability to do $p$ queries in parallel in one time-step.  An easy way to make use of this parallelism is to view the database as $p$ databases of $n/p$ bits each, and to run a separate copy of Grover's algorithm on each of those.  This finds a 1-position with high probability using $O(\sqrt{n/p})$ $p$-parallel queries, and Zalka showed that this is optimal.
Subsequently Grover and Radhakrishnan~\cite{grover&radhakrishnan:parallel} studied the problem of finding \emph{all} $k$ solutions in an $n$-bit database using $p$-parallel queries. They assume $k,p\leq\sqrt{n}$ and show that $\widetilde{\Theta}(\sqrt{nk/p\min\{k,p\}})$ $p$-parallel queries are necessary and sufficient.

\subsection{Our results}
We focus on parallel quantum algorithms in the setting of quantum query complexity. 
Consider a function $f:{\cal D}\rightarrow\01$, with ${\cal D}\subseteq[q]^n$.
For standard (sequential) query complexity, let $Q(f)$ denote the bounded-error quantum query complexity of $f$, i.e., the minimal number of queries needed among all quantum algorithms that (for every input $x\in{\cal D}$) output $f(x)$ with probability at least $2/3$.  In the $p$-parallel query model, for some integer $p\geq 1$, an algorithm can make up to $p$ quantum queries in parallel in each timestep.
In that case, we let $Q^{p\parallel}(f)$ denote the bounded-error $p$-parallel complexity of~$f$. As always in query complexity, all intermediate input-independent computation is free. For every function, we have $Q(f)/p\leq Q^{p\parallel}(f)\leq Q(f)$.

An extreme case of the parallel model is where $p$ large enough so that $Q^{p\parallel}(f)$ becomes~1; such algorithms are called ``nonadaptive,'' because 
all queries are made in parallel. Montanaro~\cite{montanaro:nonadaptiveq} showed that for total functions, such nonadaptive quantum algorithms cannot improve much over classical algorithms: every Boolean function that depends on $n$ input bits needs $p\geq n/2$ nonadaptive quantum queries for exact computation, and $p=\Omega(n)$ queries for bounded-error computation. 

Here we prove matching upper and lower bounds on the $p$-parallel complexity $Q^{p\parallel}(f)$ for a number of 
 problems: $\Theta((n/p)^{2/3})$ queries for element distinctness and $\Theta((n/p)^{k/(k+1)})$ for the \kSUM~problem for any constant $k>1$. 
Our upper bounds are obtained by parallelized quantum walk algorithms, and our lower bounds are based on a modification of the adversary lower bound method combined with some recent results by Belovs et al.\ about using so-called ``learning graphs,'' both for upper and for lower bounds~\cite{belovs:learninggraphs,belovs:nonadaptive,belovs:ksum,belovs&spalek:ksum}. The modification we need to make is surprisingly small, and technically we need to do little more than adapt recent progress on sequential algorithms to the parallel case.  Still, we feel this extension is important because: (1) 
our techniques may be useful for proving future lower bounds; (2) parallel quantum algorithms are important and yet have received little attention before; and (3) the fact that the extension is easy and natural increases our confidence that the adversary method  
is the ``right'' approach in the parallel as well as the sequential case. 

In Section~\ref{secgeneralbounds} we prove some more ``structural'' results, i.e., bounds for $Q^{p\parallel}(f)$ that hold for all total Boolean functions $f:\01^n\rightarrow\01$.  Specifically, based on earlier results in the sequential model due to Beals et al.~\cite{bbcmw:polynomialsj}, we show that if $p$ is not too large then $Q^{p\parallel}(f)$ is polynomially related to its classical deterministic $p$-parallel counterpart.
We also observe that $Q^{p\parallel}(f)\approx n/2p$ for almost all~$f$.

\section{Preliminaries}

\subsection{Sequential and parallel query complexity}
We use $[n]:=\{1,\ldots,n\}$, ${[n]\choose k}:=\{S\subseteq [n] : |S|=k\}$,
${[n]\choose \leq k}:=\{S\subseteq [n] : |S|\leq k\}$, 
and ${n\choose \leq k}:=
\sum_{s=0}^k{n\choose s}$.

We assume basic familiarity with the model of sequential quantum query algorithms~\cite{buhrman&wolf:dectreesurvey}. We will consider algorithms in the $p$-parallel quantum query model. A quantum query to an input $x\in[q]^n$ corresponds to the following unitary map on two quantum registers:
$$
\ket{i,b}\mapsto\ket{i,b+x_i}.
$$
Here the first $n$-dimensional register contains the index $i\in[n]$ of the queried element, and the value of that element is added (in $\mathbb{Z}_q$) to the contents of the second ($q$-dimensional) register. 
In order to enable an algorithm to not make a query on part of its state, we extend the previous unitary map to the case $i=0$ by $\ket{0,b}\mapsto\ket{0,b}.$
In each timestep we can make up to $p$ quantum queries in parallel by applying the map
$$
\ket{i_1,b_1,\dots,i_p,b_p}\mapsto\ket{i_1,b_1+x_{i_1},\dots,i_p,b_p+x_{i_p}}
$$ 
at unit cost. All intermediate input-independent computation is free, so the complexity of a $p$-parallel algorithm is measured solely by the number of times it applies a $p$-parallel query.

Consider a function $f:{\cal D}\rightarrow\01$, with ${\cal D}\subseteq[q]^n$.
When $p=1$ we have the standard sequential query complexity, and we let $Q_\eps(f)$ denote the quantum query complexity of $f$ with error probability~$\leq\eps$ on every input $x\in{\cal D}$. For general $p$, let $Q_\eps^{p\parallel}(f)$ be the $p$-parallel complexity of~$f$. Note that $Q_\eps(f)/p\leq Q_\eps^{p\parallel}(f)\leq Q_\eps(f)$ for every function. The exact value of the error probability $\eps$ does not matter, as long as it is a constant $<1/2$.  We usually fix $\eps=1/3$, abbreviating $Q(f)=Q_{1/3}(f)$ and $Q^{p\parallel}(f)=Q_{1/3}^{p\parallel}(f)$ as in the introduction.

We will use an extension of the adversary bound for the usual sequential (1-parallel) quantum query model. 
An \emph{adversary matrix} $\Gamma$ for~$f$ is a real-valued matrix whose rows are indexed by inputs $x\in f^{-1}(0)$ 
and whose columns are indexed by $y\in f^{-1}(1)$.\footnote{One also often sees this defined as a matrix whose rows and columns are both indexed by the set of all inputs, and that is required to be~0 on $x,y$-entries where $f(x)=f(y)$. Both definitions of an adversary matrix give the same lower bound.}
Let $\Delta_j$ be the Boolean matrix whose rows and columns are indexed by $x\in f^{-1}(0)$ and $y\in f^{-1}(1)$, such that $\Delta_j[x,y]=1$ if $x_j\neq y_j$, and $\Delta_j[x,y]=0$ otherwise.
The (negative-weights) adversary bound for~$f$ is given by: 
\begin{equation}\label{eq:advdef}
\ADV(f)=\max_\Gamma\frac{\norm{\Gamma}}{\max_{j\in[n]}\norm{\Gamma\circ \Delta_j}},
\end{equation}
where $\Gamma$ ranges over all adversary matrices for~$f$, `$\circ$' denotes entry-wise product of two matrices,
and `$\norm{\cdot}$' denotes the operator norm associated to the $\ell_2$ norm.
This lower bound (often denoted $\ADV^\pm(f)$ instead of $\ADV(f)$) was introduced by H\o yer et al.~\cite{hls:madv}, generalizing Ambainis~\cite{ambainis:lowerboundsj}. 
They showed
$$
Q_\eps(f)\geq \frac{1}{2}(1-\sqrt{\eps(1-\eps)})\ADV(f)
$$
for all~$f$.
Reichardt et al.~\cite{reichardt:tight,lmrss:stateconv} showed this is tight: $Q(f)=\Theta(\ADV(f))$ for all (total as well as partial) Boolean functions~$f$.

\subsection{Quantum walks}
We will construct and analyze our algorithms in the quantum walk framework of \cite{mnrs11}, which we now briefly describe. Given a reversible Markov process $P$ on state space $V$, and a subset $M\subset V$ of marked elements, we define three costs: the setup cost, $\mathsf{S}$, is the cost to construct a superposition $\sum_{v\in V}\sqrt{\pi_v}\ket{v}$, where $\pi_v$ is the probability of vertex~$v$ in the stationary distribution~$\pi$ of $P$; the checking cost, $\mathsf{C}$, is the cost to check if a state $v\in V$ is in $M$; and the update cost, $\mathsf{U}$, is the cost to perform the map $\ket{v}\ket{0}\mapsto \ket{v}\sum_{u\in V}\sqrt{P_{vu}}\ket{u}$, where $P_{vu}$ is the transition probability in~$P$ to go from $v$ to~$u$. Let $\delta$ be the spectral gap of~$P$, which is the difference between its largest and second-largest eigenvalue.  Let $\eps$ be a lower bound on $\sum_{v\in M}\pi_v$ whenever $M$ is nonempty.  Then we can determine if $M$ is nonempty with bounded error probability in cost 
$$
O\left(\mathsf{S}+\frac{1}{\sqrt{\eps}}\left(\frac{1}{\sqrt{\delta}}\mathsf{U}+\mathsf{C}\right)\right).
$$
If $\mathsf{S}$, $\mathsf{U}$ and $\mathsf{C}$ denote query complexities, then the above expression gives the bounded-error query complexity of the quantum walk algorithm. If they denote \emph{$p$-parallel} query complexities, the above expression gives the bounded-error $p$-parallel complexity.

\section{Lower bounds for parallel quantum query complexity}
\subsection{Adversary bound for parallel algorithms}\label{ssec:paralleladversary}

We start by extending the adversary bound for the usual sequential quantum query algorithms to $p$-parallel algorithms. For $J\subseteq[n]$, let $x_J$ be the string~$x$ restricted to the entries in~$J$. Let $\Delta_J$ be the Boolean matrix whose rows are indexed by $x\in f^{-1}(0)$ and whose columns are indexed by $y\in f^{-1}(1)$, and that has a~$1$ at position $(x,y)$ iff $x_J\neq y_J$ (i.e., $x_j\neq y_j$ for at least one $j\in J$). For $J=\emptyset$, $\Delta_J$ is the all-0 matrix. Define the following quantity:
\begin{equation}\label{eq:paralleladvdef}
\ADV^{p\parallel}(f)=\max_\Gamma\frac{\norm{\Gamma}}{\max_{J\in{[n]\choose \leq p}}\norm{\Gamma\circ \Delta_J}}.
\end{equation}
The following fact (Appendix~\ref{fact1}) implies we only need to consider sets $J\in{[n]\choose p}$ in the above definition:
$\ADV^{p\parallel}(f)$ equals 
$$
\max_\Gamma\frac{\norm{\Gamma}}{\max_{J\in{[n]\choose p}}\norm{\Gamma\circ \Delta_J}}
$$ 
up to a factor of~2. We could even use the latter as an alternative definition of $\ADV^{p\parallel}(f)$.

\begin{fact}\label{optp}
For every set $J\subseteq K\subseteq[n]$, we have
$\norm{\Gamma\circ \Delta_J}\leq 2\norm{\Gamma\circ \Delta_{K}}$.
\end{fact}
We now show that, just like in the sequential case, the adversary bound characterizes the quantum query complexity also in the $p$-parallel case:

\begin{theorem}\label{thm:par-adv}
For every $f:{\cal D}\rightarrow\01$ and ${\cal D}\subseteq[q]^n$, $\displaystyle Q^{p\parallel}(f) = \Theta(\ADV^{p\parallel}(f)).$
\end{theorem}

\begin{proof}
In order to derive $p$-parallel lower bounds from sequential lower bounds, observe that we can make a bijection between input $x\in [q]^n$ and a larger string $X$ indexed by all sets $J\in{[n]\choose \leq p}$, such that $X_{J}=(x_{j})_{j\in J}$.  
That is, each index $J$ of~$X$ corresponds to up to~$p$ indices~$j$ of~$x$.
We now define a new function $F:{\cal D}'\rightarrow\01$,
where ${\cal D}'$ is the set of $X$ as above, in 1-to-1 correspondence with the elements of $x\in\cal D$, and $F(X)$ is defined as~$f(x)$. %
\footnote{Note that for $p>1$ the new function $F$ is partial, even if the underlying $f$ is total. A similar translation from parallel to sequential complexity was used by Grover and Radhakrishnan~\cite[Section~2]{grover&radhakrishnan:parallel} for the special case of searching a database.}
One query to $X$ can be simulated by $p$ parallel queries to $x$, and vice versa, so we have
$Q^{p\parallel}(f)=Q(F)$.
As mentioned at the end of Section~\ref{ssec:paralleladversary}, we have $Q(F)=\Theta(\ADV(F))$. Now Eq.~\eqref{eq:advdef} applied to~$F$ gives the claimed lower bound of Eq.~\eqref{eq:paralleladvdef} on $Q^{p\parallel}(f)$.
\end{proof}

Sometimes we can even use the same adversary matrix $\Gamma$ to obtain optimal lower bounds for~$F$ and~$f$. 
An example is the $n$-bit OR-function. Let $\Gamma$ be the all-ones $1\times n$ matrix, with the row corresponding to input~$0^n$ and the columns indexed by all weight-1 inputs. Then $\norm{\Gamma}=\sqrt{n}$ and $\norm{\Gamma\circ\Delta_j}=1$ for all $j\in[n]$, and hence $Q(\OR)=
\Omega(\sqrt{n})$. To get $p$-parallel lower bounds, we define a new function $F:X\mapsto\01$ as in the proof of Theorem~\ref{thm:par-adv}. We can use the same~$\Gamma$, with the $n$ columns still indexed by the weight-1 inputs to $f$ (which induce 1-inputs to $F$). Now $J$ ranges over subsets of $[n]$ of size at most $p$, and $\Delta_J$ will be the matrix whose $(x,y)$-entry is~1 if there is at least one $j\in J$ such that $x_j\neq y_j$. Note that $\norm{\Gamma\circ\Delta_J}=\sqrt{|J|}$ for all $J$.
Hence $Q^{p\parallel}(\OR)=\Omega(\ADV(F))=\Omega(\sqrt{n/p})$. This is optimal and was already proved (in a different way) by Zalka~\cite[Section~4]{zalka:grover}.

\subsection{Belovs's learning graph approach}\label{sec:belovs}

Recently Belovs~\cite{belovs:learninggraphs} gave a new approach to designing (sequential) quantum algorithms, via the optimality of the adversary method. He introduced the model of \emph{learning graphs} to prove upper bounds on the adversary bound, and hence upper bounds on quantum query complexity.  
We state it here for \emph{certificate structures}.
We define these below, slightly simpler and less general than Definitions~1 and~3 of Belovs and Rosmanis~\cite{belovs:nonadaptive} (for us $M$ denotes a minimal certificate, while in~\cite{belovs:nonadaptive} it denotes the set of supersets of a minimal certificate).

\begin{definition}
Let $\cal C$ be a set of incomparable subsets of~$[n]$.
We say $\cal C$ is a \emph{1-certificate structure} for a function $f:{\cal D}\rightarrow\01$, with ${\cal D}\subseteq[q]^n$,  
if  for every $x\in f^{-1}(1)$ there exists an $M\in{\cal C}$ such that for all $y\in{\cal D}$, $y_M=x_M$ implies $f(y)=1$.
We say $\cal C$ is \emph{$k$-bounded} if $|M|\leq k$ for all $M\in{\cal C}$.
\end{definition}

The 
learning graph complexity of $\cal C$ is defined in the following  
in its primal formulation as a minimization problem (we will see an equivalent dual formulation soon).
Let ${\cal E}=\{(S,j): S\subseteq[n], j\in[n]\backslash S\}$. For $e=(S,j)\in {\cal E}$, we use $s(e)=S$ and $t(e)=S\cup\{j\}$.
\begin{align}
&\LGC({\cal C})= \min \sqrt{\textstyle\sum_{e\in{\cal E}} w_e} \quad\mbox{such that}\label{fig:lg}\\
 &\quad \sum_{e\in{\cal E}}\frac{\theta_e(M)^2}{w_e}\leq 1 & \mbox{for all }M\in{\cal C}\label{eq:energy}\\
&\quad \sum_{e\in{\cal E}: t(e)=S} \theta_e(M)=\sum_{e\in{\cal E}: s(e)=S} \theta_e(M) & \mbox{for all } M\in{\cal C},\emptyset\neq S\subseteq[n],M\not\subseteq S \label{eq:flow}\\
&\quad \sum_{e\in{\cal E}: s(e)=\emptyset} \theta_e(M)=1 & \mbox{for all }M\in{\cal C} \label{eq:unit-flow}\\
&\quad \theta_e(M)\in\mathbb{R}, w_e\geq 0 & \mbox{for all }e\in{\cal E}\mbox{ and }M\in{\cal C}\label{eq:flow2}
\end{align}
For each $M$,  $\theta_e(M)$ is a \emph{flow} from $\emptyset$ to $M$ on the graph with vertices $\{S\subseteq [n]\}$ and edges $\{\{S,S\cup\{j\}\}:(S,j)\in \cal{E}\}$ if $\theta_e(M)$ satisfies condition~\eqref{eq:flow}.
Moreover, $\theta_e(M)$ is a \emph{unit flow} if it also satisfies condition~\eqref{eq:unit-flow}.

Belovs showed that the learning graph complexity of $\cal C$ is an upper bound on $\ADV(f)$, and hence on $Q(f)$, for any function~$f$ with certificate structure $\cal C$. This bound is not always optimal, since it only depends on the certificate structure of $f$. For example, $k$-distinctness has
quantum query complexity $o(n^{3/4})$ 
even though it has the same 1-certificate structure as $k$-sum, whose quantum query complexity is $\Theta(n^{k/(k+1)})$~\cite{belovs:ksum,belovs&spalek:ksum}.
However, Belovs and Rosmanis~\cite{belovs:nonadaptive} proved that for a special class of functions, 
it turns out the upper bound LGC(${\cal C})$ is optimal.
\begin{definition}
{An \emph{orthogonal array} of length $k$ is a set $T\subseteq[q]^k$, such that for every $i\in[k]$ and every $x_1,\ldots,x_{i-1},x_{i+1},\ldots,x_k$ there exists exactly one $x_i\in[q]$ such that $(x_1,\ldots,x_k)\in T$.}
\end{definition}

\begin{theorem}[Belovs-Rosmanis]\label{thq=lgc}
Let $\cal C$ be a $k$-bounded 1-certificate structure for some constant $k$, $q\geq 2|{\cal C}|$, and let each $M\in{\cal C}$ be equipped with an orthogonal array $T_M$ of length $|M|$. Define a Boolean function $f:[q]^n\rightarrow\01$ by: $f(x)=1$ iff there exists an $M\in{\cal C}$ such that $x_M\in T_M$. Then $Q(f)=\Theta(\LGC({\cal C}))$.
\end{theorem}

For example, the element distinctness problem \ED\ on input $x\in[q]^n$ is 
 induced by the 2-bounded 1-certificate structure ${\cal C}={[n]\choose 2}$, equipped with associated orthogonal arrays $T_{\{i,j\}}=\{(v,v):v\in [q]\}$.
Hence $Q(\ED)=\Theta(\LGC({\cal C}))$.

Belovs and Rosmanis~\cite{belovs:nonadaptive} 
show that an equivalent dual definition of the learning graph complexity as a maximization problem is the following:

\begin{align}
\LGC({\cal C})=& \max \sqrt{\textstyle\sum_{M\in{\cal C}} \alpha_\emptyset(M)^2}  \label{eqlgcdual} \\
\mbox{s.t. } & \textstyle\sum_{M\in{\cal C}}(\alpha_{s(e)}(M)-\alpha_{t(e)}(M))^2\leq 1 & \mbox{for all }e\in{\cal E} \label{eqlgcdualconstraint} \\ 
& \alpha_S(M)=0 & \mbox{whenever }M\subseteq S\nonumber\\
& \alpha_S(M)\in\mathbb{R} & \hspace{-10pt} \mbox{for all } S\subseteq [n]\mbox{ and }M\in{\cal C}\nonumber
\end{align}
In particular, that means we can prove \emph{lower} bounds on $\LGC({\cal C})$ (and hence, for the functions described in Theorem~\ref{thq=lgc}, on $Q(f)$) by exhibiting a feasible solution $\{\alpha_S(M)\}$ for this maximization problem and calculating its objective value.

Before stating a similar result for $p$-parallel query complexity, we first adapt learning graphs.
Edges, which previously were of type $e=(S,j)$ with $S\subseteq [n]$ and $j\in[n]\setminus S$, are now of type $e=(S,J)$ with $S\subseteq [n]$, $J\subseteq[n]\setminus S$ and $|J|\leq p$.

\begin{definition}
The \emph{$p$-parallel 
learning graph complexity} $\LGC^{p\parallel}({\cal C})$ of $\cal C$ is defined as $\LGC({\cal C})$
where we replace the edge set $\mathcal{E}$ with ${\cal E}_p=\{(S,J):S\subseteq [n], J\subseteq [n]\setminus S, |J|\leq p\}$.  
Its dual 
is analogous. In particular, we replace constraint~\eqref{eqlgcdualconstraint} by 
$$
\sum_{M\in{\cal C}}(\alpha_{s(e)}(M)-\alpha_{t(e)}(M))^2\leq 1\mbox{~for all }e=(S,J)\in{\cal E}_p,
$$ 
where $s(e)=S$ and $t(e)=S\cup J$.  We call this modified constraint ``parallel-\eqref{eqlgcdualconstraint}.''
\end{definition}

As in the special case of $p=1$, the $p$-parallel learning graph complexity of $\mathcal{C}$ provides an upper bound on $\ADV^{p\parallel}(f)$, and hence on $Q^{p\parallel}(f)$, for any function $f$ having that same certificate structure. The proof is given in Appendix~\ref{app:parallel-cert}.

\begin{lemma}\label{lem:parallel-cert}
Let $\mathcal{C}$ be a certificate structure for $f$.
Then $\ADV^{p\parallel}(f)\leq \LGC^{p\parallel}(\mathcal{C})$.
\end{lemma}

We now generalize Theorem~\ref{thq=lgc} to the $p$-parallel case. The proof of Theorem \ref{thq=p-lgc} can be found in Appendix \ref{app:GammaDeltaJ}. It is an adaptation of the proof of \cite[Theorem~5]{belovs:nonadaptive}.
\begin{theorem}\label{thq=p-lgc}
Let $\cal C$ be a $k$-bounded 1-certificate structure for some constant $k$, $q\geq 2|{\cal C}|$, and let each $M\in{\cal C}$ be equipped with an orthogonal array $T_M$ of length $|M|$. Define a Boolean function $f:[q]^n\rightarrow\01$ as follows: $f(x)=1$ iff there exists an $M\in{\cal C}$ such that $x_M\in T_M$. Then $Q^{p\parallel}(f)=\Theta(\LGC^{p\parallel}({\cal C}))$.
\end{theorem}

\section{Parallel quantum query complexity of specific functions}
\subsection{Algorithms}\label{secupperbounds}

In this section we give upper bounds for element distinctness and $k$-sum in the $p$-parallel quantum query model, by way of quantum walk algorithms. 

Our $p$-parallel algorithm for element distinctness is based on the sequential query algorithm for element distinctness of Ambainis~\cite{ambainis:edj}. Ambainis's algorithm uses a quantum walk on a Johnson graph, $J(n,r)$, which has vertex set $V=\{S\subseteq[n]:|S|=r\}$ and edge set $\{\{S,S'\}\subseteq V:|S\setminus S'|=1\}$.
Each state $S\in V$ represents a set of queried indices. The algorithm seeks a  state $S$ containing $(i,x_i)$ and $(j,x_j)$ such that $i\neq j$ and $x_i=x_j$. Such a state is said to be \emph{marked}.

\begin{theorem}\label{th:ed-u-bd}
Element distinctness on $[q]^n$ has $ Q^{p\parallel}(\ED)=O((n/p)^{2/3})$.
\end{theorem}

\begin{proof}
We modify Ambainis's quantum walk algorithm slightly. 
Consider a walk $J(n,r/p)^p$, on $p$ copies of the Johnson graph $J(n,r/p)$.
Vertices are $p$-tuples $(S_1,S_2,\ldots,S_p)$ where, for each $i\in [p]$, $S_i\subseteq [n]$ and $|S_i|=r/p$.
Two vertices $(S_1,S_2,\ldots,S_p)$ and $(S'_1,S'_2,\ldots,S'_p)$ are adjacent if, for each $i\in [p]$, $|S_i\setminus S'_i|=1$.
We call a state $(S_1,S_2,\ldots,S_p)$ \emph{marked} if there are $j,j'\in\bigcup_{i=1}^p S_i$ such that $x_j=x_{j'}$.
Since the stationary distribution is $\mu^p$, where $\mu$ is the uniform distribution on $\binom{[n]}{r/p}$, 
the probability that a state is marked is at least $\eps=\Omega(r^2/n^2)$. 

The setup cost is only $\mathsf{S}=O(r/p)$ $p$-parallel queries, since it suffices to query $r$ elements in the initial superposition over all states. 
Similarly, the update requires that we query and unquery $p$ elements, but we can accomplish this in two $p$-parallel queries, so~$\mathsf{U}=O(1)$. Also, $\mathsf{C}=0$.
Finally, the eigenvalues of the product of $p$ copies of a graph are exactly the products of $p$ eigenvalues of that graph. Hence if the largest eigenvalue of a graph is 1 and the second-largest is $1-\delta$, then the same will be true for the product graph.  Accordingly, the spectral gap~$\delta$ of $p$ copies of $J(n,r/p)$ is exactly the spectral gap of one copy of $J(n,r/p)$, which is $\Omega(p/r)$.
We can now upper bound the $p$-parallel query complexity of element distinctness as
$$
O\left(\mathsf{S}+\frac{1}{\sqrt{\eps}}\left(\frac{1}{\sqrt{\delta}}\mathsf{U}+\mathsf{C}\right)\right)  = 
O\left(\frac{r}{p}+\frac{n}{r}\left(\sqrt{\frac{r}{p}}\right)\right) 
 =  O\left(\frac{r}{p}+\frac{n}{\sqrt{rp}}\right).
$$
Setting $r$ to the optimal $n^{2/3}p^{1/3}$ gives an upper bound of $O((n/p)^{2/3})$. 
\end{proof}

\noindent It is easy to generalize our element distinctness upper bound to $k$-sum:

\begin{theorem}\label{thm:ksum-algo}
$k$-sum  on $[q]^n$ has $ Q^{p\parallel}(\kSUM)=O((n/p)^{k/(k+1)})$.
\end{theorem}
\begin{proof}
Again, we walk on $p$ copies of $J(n,r/p)$, but now we consider a state $(S_1,S_2,\ldots,S_p)$ marked if there are distinct indices $i_1,\ldots,i_k\in \bigcup_{i=1}^p S_i$ such that $\sum_{j=1}^k x_{i_j}=0\pmod q$. The fraction of marked states in a $1$-instance is $\eps=\Omega(r^k/n^k)$. All other parameters are as in Theorem~\ref{th:ed-u-bd}. We get the following upper bound for $k$-sum:
\begin{eqnarray*}
O\left(\mathsf{S}+\frac{1}{\sqrt{\eps}}\left(\frac{1}{\sqrt{\delta}}\mathsf{U}+\mathsf{C}\right)\right)  = 
O\left(\frac{r}{p}+\frac{n^{k/2}}{r^{k/2}}\left(\sqrt{\frac{r}{p}}\right)\right)
 =  O\left(\frac{r}{p}+\frac{n^{k/2}}{r^{(k-1)/2}\sqrt{p}}\right).
\end{eqnarray*}
Setting $r=n^{k/(k+1)}p^{1/(k+1)}$ gives $O((n/p)^{k/(k+1)})$. 
\end{proof}

\subsection{Lower bounds}

We now use the ideas from Section~\ref{sec:belovs} to prove $p$-parallel lower bounds for \ED{} and \kSUM, matching our upper bounds 
if the alphabet size~$q$ is sufficiently large. Our proofs are generalizations of the sequential lower bounds in~\cite[Section~4]{belovs:nonadaptive}.

\begin{theorem}\label{thEDparallelLB}
For $q\geq 2{n\choose 2}$, element distinctness on $[q]^n$ has $Q^{p\parallel}(\ED)=\Omega((n/p)^{2/3})$.
\end{theorem}

\begin{proof}
Recall that element distinctness is induced by the 1-certificate structure ${\cal C}={[n]\choose 2}$, equipped with associated orthogonal arrays $T_{\{i,j\}}=\{(v,v):v\in [q]\}$.
By Theorem~\ref{thq=p-lgc}, it suffices to prove the lower bound on the $p$-parallel learning graph complexity of $\ED$.
For this, it suffices to exhibit a feasible solution to the parallel version of dual~\eqref{eqlgcdual} and to lower bound its objective function.  
Note that the elements of ${\cal E}_p$ are now of the form $(S,J)$, where $S\subseteq [n]$ and $J\subseteq[n]\setminus S$ with $|J|\leq p$.
Define 
$$ 
\alpha_j=\frac{1}{2n} \max( (n/p)^{2/3}- j/p, 0), \quad\mbox{and}\quad \alpha_S(M)=\left\{\begin{array}{ll} 
0 & \mbox{if }M\subseteq S\\
\alpha_{|S|} & \mbox{otherwise.}\end{array}\right.
$$
To show that this is a feasible solution, the only constraint we need to verify is parallel-\eqref{eqlgcdualconstraint}.
Fix $S\subseteq[n]$ of some size~$s$, and a set $J\subseteq [n]\setminus S$ with $|J|\leq p$.
Let $L$ denote the left-hand side of parallel-\eqref{eqlgcdualconstraint}, which is a sum over all ${n\choose 2}$ certificates $M\in{\cal C}$. 
With respect to $e=(S,J)$, there are four kinds of $M=\{i,j\}$:
\begin{enumerate}
\item $i,j\in S$. Then $\alpha_{t(e)}(M)=\alpha_{s(e)}(M)=0$, so these~$M$ contribute~0 to $L$.
\item $i\in S,j\in J$. There are $s|J|\leq sp$ such~$M$, and each contributes $\alpha_s^2$ to~$L$ because $\alpha_{s(e)}(M)=\alpha_s$ and $\alpha_{t(e)}(M)=0$.
\item $i,j\not\in S$, $i,j\in J$. There are ${|J|\choose 2}\leq {p\choose 2}$ such~$M$, each contributes $\alpha_s^2$ to~$L$.
\item $i$ and/or $j\not\in S\cup J$. There are ${n\choose 2}-{s+|J|\choose 2}\leq n^2$ such~$M$, each contributes $|\alpha_s-\alpha_{s+|J|}|^2$ to~$L$.
\end{enumerate}
Hence, using $\alpha_s=0$ if $s\geq n^{2/3}p^{1/3}$; $\alpha_s\leq\alpha_0=\frac{1}{2p^{2/3}n^{1/3}}$;
and $|\alpha_s-\alpha_{s+|J|}|^2\leq 1/4n^2$, we can establish constraint parallel-\eqref{eqlgcdualconstraint}:
$$
L\leq 
\left(sp+{p\choose 2}\right)\alpha_s^2 + n^2|\alpha_s-\alpha_{s+|J|}|^2 \leq p(n^{2/3}p^{1/3}+p/2)\frac{1}{4p^{4/3}n^{2/3}} + n^2\frac{1}{4n^2}\leq 1.
$$
Hence our solution is feasible.
Its objective value is $\sqrt{{n\choose 2}\alpha_0^2}=\Omega((n/p)^{2/3})$. 
\end{proof}

\noindent The lower bound proof for $\kSUM$ is similar.
Here we use certificate structure ${\cal C}={[n]\choose k}$ with the orthogonal array $T=\{(v_1,\ldots,v_k) : \sum_{i=1}^k v_i=0 \mod q\}$, which induces \kSUM. In Appendix \ref{appksumparallelLB}, we show that the following solution is feasible for $\LGC^{p\parallel}(\cal C)$:
$$
\alpha_j=\frac{1}{2n^{k/2}}\max((n/p)^{k/(k+1)} - j/p, 0)\quad\mbox{and}\quad \alpha_S(M)=\left\{\begin{array}{ll}
0 &\mbox{if }M\subseteq S\\
\alpha_{|S|} & \mbox{otherwise;}\end{array}\right.
$$
Since its objective value is $\sqrt{{n \choose k} \alpha_0^2}=\Omega\left((n/p)^{k/(k+1)}\right)$, we obtain

\begin{theorem}\label{thksumparallelLB}
For $q\geq 2{n\choose k}$, $\kSUM$ on $[q]^n$ has $Q^{p\parallel}(\kSUM)=\Omega\left((n/p)^{k/(k+1)}\right)$.
\end{theorem}

\section{Some general bounds}\label{secgeneralbounds}

In this section we will relate quantum and classical $p$-parallel complexity.
For the sequential model ($p=1$) it is known that quantum bounded-error query complexity is no more than a 6th power less than classical deterministic complexity, for all total Boolean functions~\cite{bbcmw:polynomialsj}.  Here we will see to what extent we can prove a similar result for the $p$-parallel model.

We start with a few definitions, referring to~\cite{buhrman&wolf:dectreesurvey} for more details.
Let $f:\01^n\rightarrow\01$ be a total Boolean function.
For $b\in\01$, a \emph{$b$-certificate} for~$f$ is an assignment~$C:S\rightarrow\01$ to a subset $S$ of the $n$ variables, such that $f(x)=b$ whenever $x$ is consistent with $C$.
The \emph{size} of $C$ is $|S|$.
The \emph{certificate complexity $C_x(f)$ of $f$ on $x$} is the size of
a smallest $f(x)$-certificate that is consistent with $x$.
The \emph{certificate complexity} of $f$ is $C(f)=\max_x C_x(f)$.
The \emph{$1$-certificate complexity} of $f$ is
$C^{(1)}(f)=\max_{\{x:f(x)=1\}}C_x(f)$.
Given an input $x\in\01^n$ and subset $B\subseteq[n]$ of indices of variables, let $x^B$ denote the $n$-bit input obtained from~$x$ by negating all bits~$x_i$ whose index~$i$ is in~$B$.
The \emph{block sensitivity} $bs(f,x)$ of $f$ at input~$x$, is the maximal integer~$k$ such that there exist disjoint sets $B_1,\ldots,B_k$ satisfying $f(x)\neq f(x^{B_i})$ for all $i\in[k]$.
The \emph{block sensitivity} of $f$ is $bs(f)=\max_x bs(f,x)$.
Nisan~\cite{nisan:pram&dt} proved that 
\begin{equation}\label{eq:cbs2}
bs(f)\leq C(f)\leq bs(f)^2.
\end{equation}
Via a standard reduction~\cite{nisan&szegedy:degree}, Zalka's $\Theta(\sqrt{n/p})$ bound for OR implies:

\begin{theorem}\label{th:bsplowerbound}
For every $f:\01^n\rightarrow\01$, $Q^{p\parallel}(f)=\Omega(\sqrt{bs(f)/p}).$
\end{theorem}

\noindent We now prove a general upper bound on deterministic $p$-parallel complexity:

\begin{theorem}\label{th:DCbs}
For every $f:\01^n\rightarrow\01$,
$D^{p\parallel}(f) \leq \ceil{C^{(1)}(f)/p}bs(f).$
\end{theorem}

\begin{proof}
Beals et al.~\cite[Lemma~5.3]{bbcmw:polynomialsj} give a deterministic decision tree for~$f$ that runs for at most $bs(f)$ rounds, and in each round queries all variables of a 1-certificate, substituting their values into the function. They show this reduces the function to a constant. By parallelizing the querying of the certificate we can implement every round using $\ceil{C^{(1)}(f)/p}$ $p$-parallel steps.
\end{proof}

\noindent
$D^{p\parallel}(f)$ and $Q^{p\parallel}(f)$ are polynomially related if $p$ is not too big:

\begin{theorem}\label{th:DvsQparallel}
For every $f:\01^n\rightarrow\01$, $c>1$, $p\leq bs(f)^{1/c}$, we have
$\displaystyle D^{p\parallel}(f)\leq O(Q^{p\parallel}(f)^{6+4/(c-1)}).$
\end{theorem}

\begin{proof}
We can assume $C(f)=C^{(1)}(f)$ (else consider~$1-f$).
By Eq.~\eqref{eq:cbs2} we have $p\leq bs(f)^{1/c}\leq C^{(1)}(f)$. 
We also have $C^{(1)}(f) \leq bs(f)^2$.
The assumption on $p$ is equivalent to $p\leq (bs(f)/p)^{1/(c-1)}$. 
Using Theorems~\ref{th:bsplowerbound} and~\ref{th:DCbs}, we obtain
\begin{align*}
D^{p\parallel}(f) & \leq  \ceil{C^{(1)}(f)/p}bs(f) \leq O(bs(f)^3/p)=O((bs(f)/p)^3 p^2)\\
& \leq  O((bs(f)/p)^{3+2/(c-1)})\leq O(Q^{p\parallel}(f)^{6+4/(c-1)}). 
\end{align*}
\end{proof}

For example, if $p\leq bs(f)^{1/3}$ then $Q^{p\parallel}(f)$ is at most an 8th power smaller than $D^{p\parallel}(f)$. Whether superpolynomial gaps exist for large~$p$ remains open.

We end with an observation about random functions.
Van Dam~\cite{dam:oracle} showed that an $n$-bit input string~$x$ can be recovered with high probability using $n/2+O(\sqrt{n})$ quantum queries, hence $Q(f)\leq n/2+O(\sqrt{n})$ for all $f:\01^n\rightarrow\01$.
His algorithm already applies its queries in parallel, so allows us to compute~$x$ using roughly $n/2p$ $p$-parallel quantum queries (see Appendix~\ref{app:parallelinter}).
Ambainis et al.~\cite{absw:nover2} proved an essentially optimal lower bound for random functions: almost all~$f$ have $Q(f)\geq (1/2-o(1))n$. Since trivially $Q(f)\leq pQ^{p\parallel}(f)$, we obtain the $p$-parallel lower bound $Q^{p\parallel}(f)\geq (1/2-o(1))n/p$ for almost all~$f$. 
So for this type of ``quantum oracle interrogation,'' parallelization gives the optimal factor-$p$ speed-up for almost all Boolean functions.

\begin{corollary}
For all $p\leq n$, almost all $f:\01^n\rightarrow\01$ satisfy $Q^{p\parallel}(f)=(1/2\pm o(1))n/p$.
\end{corollary}
For $p=n/2+O(\sqrt{n})$, one $p$-parallel query suffices.

\section{Conclusion and future work}

This paper is the first to systematically study the power and limitations of parallelism
for quantum query algorithms. It is motivated in particular by the need to reduce overall 
computing time when running quantum algorithms on hardware with quickly decohering quantum bits.

We leave open many interesting questions for future work, for example:
\begin{itemize}
\item There are many other computational problems whose $p$-parallel complexity is unknown, for example
finding a triangle in a graph or deciding whether two given matrices multiply to a third one.  For both of these problems, however, even the sequential quantum query complexity is still open.
\item We suspect Theorem~\ref{th:DvsQparallel} is non-optimal, and
conjecture that $D^{p\parallel}(f)$ and $Q^{p\parallel}(f)$ are polynomially related for large~$p$ as well. Montanaro's result~\cite{montanaro:nonadaptiveq} about the weakness of maximally parallel quantum algorithms is evidence for this. Even for the sequential model ($p=1$) the correct bound is open; the best known bound is a 6th power~\cite{bbcmw:polynomialsj} but the correct power may well be~2.
\item Can we find relations with quantum communication complexity? Nonadaptive quantum query algorithms induce one-way communication protocols, while fully adaptive ones induce protocols that are very interactive. Our $p$-parallel algorithms would sit somewhere in between.
\end{itemize}

\noindent
{\bf Acknowledgment.}
We thank J\'er\'emie Roland for helpful discussions.

\bibliographystyle{plain}
\bibliography{qc}

\appendix

\section{Proof of Fact~\ref{optp}}\label{fact1}

We use the $\gamma_2$-norm for matrices, which is defined as 
$$
\gamma_2(A)=\min_{X,Y:A=XY} r(X)c(Y),
$$
where $r(X)$ denotes the maximum squared length among the rows of~$X$, and
$c(Y)$ denotes the maximum squared length among the columns of~$Y$.
Note that the identity and the all-1 matrix both have $\gamma_2$-norm equal to~1 (the latter by taking $X$ and $Y$ to be the all-1 row and columns, respectively). Also, $\gamma_2(A\otimes B)=\gamma_2(A)\gamma_2(B)$.
Since $\Delta_J$ can be written as the all-1 matrix of the appropriate dimensions, minus identity tensored with a smaller all-1 matrix, the triangle inequality implies $\gamma_2(\Delta_J)\leq 2$.
The $\gamma_2$-norm satisfies $\norm{A\circ B}\leq \norm{A}\gamma_2(B)$ by~\cite[Lemma~A.1]{lmrss:stateconv}.
Observe that $\Gamma\circ \Delta_J= (\Gamma\circ \Delta_K)\circ \Delta_J$. Hence we have
\begin{align*}
\norm{\Gamma\circ \Delta_J}&=\norm{(\Gamma\circ \Delta_K)\circ \Delta_J} \leq
\norm{\Gamma\circ \Delta_K}\gamma_2(\Delta_J)\leq 2\norm{\Gamma\circ \Delta_K}.
\end{align*}

\section{Proof of Lemma~\ref{lem:parallel-cert}}\label{app:parallel-cert}

The proof is a straightforward adaptation of the proof of \cite[Theorem~9]{bl11}, but we repeat it here for completeness. Let $\{w_{S,J}:(S,J)\in \mathcal{E}_p\}$ and $\{\theta_{S,J}(M):(S,J)\in\mathcal{E}_p, M\in\cal{C}\}$ be an optimal solution to the primal formulation of $\LGC^{p\parallel}(\cal{C})$. 

We will use this solution to construct a feasible solution to the dual expression of our $p$-parallel adversary of Eq.~\eqref{eq:paralleladvdef}, which is the following:
\begin{align}
\ADV^{p\parallel}(f) = & \min_{\{\ket{u_{x,J}}\}} \sqrt{\max_{x\in [q]^n}\sum_{J\in\binom{[n]}{\leq p}}\norm{\ket{u_{x,J}}}^2} & \label{eq:dual-obj} \\ 
\mbox{s.t. } &  \ket{u_{x,J}}\in \mathbb{C}^k & \mbox{ for all } x\in [q]^n,J\in \binom{[n]}{\leq p} \nonumber \\
& \sum_{J:x_J\neq y_J}\langle{u_{x,J}}\vert{u_{y,J}}\rangle=1  & \mbox{ for all }x\in f^{-1}(1), y\in f^{-1}(0) \nonumber
\end{align}
The dimension~$k$ of the vectors $\ket{u_{x,J}}$ can be anything, and is implicitly minimized over.

For each $x\in f^{-1}(1)$, let $M_x\in\cal{C}$ be such that for every $y\in [q]^n$,  $x_{M_x}=y_{M_x}$ implies $f(y)=1$. 
For every $x\in {\cal D}$ and $J\in \binom{[n]}{\leq p}$, define the following state in span$\{\ket{S}\ket{\alpha}:S\subseteq[n],\alpha\in[q]^S\}$:
$$
\ket{u_{x,J}}:=\left\{\begin{array}{ll}
\sum_{S\subseteq [n]\setminus J} \sqrt{w_{S,J}} \ket{S,x_S}& \mbox{if }f(x)=0\\
\sum_{S\subseteq [n]\setminus J} \frac{\theta_{S,J}(M_x)}{\sqrt{w_{S,J}}} \ket{S,x_S}& \mbox{if }f(x)=1\\
\end{array}\right.
$$
We now verify that $\{\ket{u_{x,J}}\}_{x,J}$ is a feasible solution to the dual formulation of $\ADV^{p\parallel}(f)$:
\begin{eqnarray}
\sum_{J\in\binom{[n]}{\leq p}:x_J\neq y_J}\inp{u_{x,J}}{u_{y,J}} & = & \sum_{J\in\binom{[n]}{\leq p}:x_J\neq y_J}\sum_{S\subseteq [n]\setminus J:x_S=y_S}\frac{\theta_{S,J}(M_x)}{\sqrt{w_{S,J}}}\sqrt{w_{S,J}}\\
&=& \sum_{S\subseteq[n]:x_S=y_S}\sum_{J\in \binom{[n]\setminus S}{\leq p}: x_J\neq y_J}\theta_{S,J}(M_x).\label{eq:cut}
\end{eqnarray}
To see that this expression is equal to~$1$, we need only notice that Eq.~\eqref{eq:cut} is the sum of the flow on all edges across the cut induced by the set $\{S\subseteq [n]:x_S=y_S\}$, and the total flow across a cut is always~$1$, since $\theta(M_x)$ is a unit flow. Thus the constraint from~\eqref{eq:dual-obj} is satisfied and $\{\ket{u_{x,J}}\}_{x,J}$ is a feasible solution. 

We can now bound $\ADV^{p\parallel}(f)$ by the objective value of the feasible solution $\{\ket{u_{x,J}}\}_{x,J}$. First note that for any $x\in f^{-1}(1)$, by constraint~\eqref{eq:energy}, we have:
$$
\sum_{J\in \binom{[n]}{\leq p}}\norm{\ket{u_{x,J}}}^2 =\sum_{J\in \binom{[n]}{\leq p}}\sum_{S\subseteq [n]\setminus J}\frac{\theta_{S,J}(M_x)^2}{w_{S,J}}\leq 1.
$$
Second, for any $x\in f^{-1}(0)$ we have
$$
\sum_{J\in \binom{[n]}{\leq p}}\norm{\ket{u_{x,J}}}^2 =\sum_{J\in\binom{[n]}{\leq p}}\sum_{S\subseteq [n]\setminus J}w_{S,J}.
$$
We can therefore bound the objective value as:
\begin{eqnarray*}
\ADV^{p\parallel}(f)&\leq & \sqrt{\max_{x\in [q]^n}\sum_{J\in\binom{[n]}{\leq p}}\norm{\ket{u_{x,J}}}^2} 
\leq \sqrt{\max\left\{1,\sum_{J\in\binom{[n]}{\leq p}}\sum_{S\subseteq [n]\setminus J}w_{S,J}\right\}}\\
&\leq & \sqrt{\sum_{e\in {\cal E}_p}w_e} = \LGC^{p\parallel}(\cal{C}),
\end{eqnarray*}
where $\sum_e w_e\geq 1$ follows from the Cauchy-Schwarz inequality and the fact that $\sum_e\frac{\theta_e(M_x)^2}{w_e}\leq 1$, as follows:
$$
1=(\sum_e \theta_e(M_x))^2=\left(\sum_e \frac{\theta_e(M_x)}{\sqrt{w_e}}\right)\leq \sum_e \frac{\theta_e(M_x)^2}{w_e}\sum_e w_e\leq \sum_e w_e.
$$

\section{Proof of Theorem~\ref{thq=p-lgc}}\label{app:GammaDeltaJ}

For the upper bound, we immediately obtain $Q^{p\parallel}(f)=O(\LGC^{p\parallel}(\mathcal{C}))$ by Theorem \ref{thm:par-adv} and Lemma~\ref{lem:parallel-cert}.


For the lower bound, our proof will be similar to that of~\cite[Theorem~5]{belovs:nonadaptive}, and we will omit parts of the proof that are identical to theirs. Just as in~\cite[Theorem~5]{belovs:nonadaptive}, our proof will start with an optimal feasible solution $\{\alpha_S(M)\}_{M\in{\cal C}, S\subseteq [n]}$ to the dual~\eqref{eqlgcdual}.
{Therefore $\LGC^{p\parallel}(\mathcal{C})=\sqrt{\sum_{M\in{\cal C}} \alpha_{\emptyset}(M)^2}$
and moreover $\sum_{M\in{\cal C}} (\alpha_S(M)-\alpha_{S\cup J}(M))^2\leq 1$.}
Then we will construct an adversary matrix $\Gamma$ for $f$ such that
$\norm{\Gamma}\geq\sqrt{\frac{1}{2}\sum_{M\in{\cal C}} \alpha_{\emptyset}(M)^2}$ (as proven in~\cite{belovs:nonadaptive}) and for every $J\subseteq [n]$, 
$$
\norm{\Gamma\circ\Delta_J}\leq 2\max_{S\subseteq [n]} \sqrt{\sum_{M\in{\cal C}} (\alpha_S(M)-\alpha_{S\cup J}(M))^2}$$ 
(proven below).
{Then Theorem~\ref{thm:par-adv} will imply $Q^{p\parallel}(f)=\Omega(\sqrt{\sum_{M\in{\cal C}} \alpha_{\emptyset}(M)^2})=\Omega(\LGC^{p\parallel}(\mathcal{C}))$.}


First we use a variation of the adversary bound from~\cite{belovs&spalek:ksum} that allows the duplication of row and column indices.
Concretely, rows and columns of $\Gamma$ are now indexed by $(x,a)$ and $(y,b)$, respectively, where $x\in f^{-1}(1)$, $y\in f^{-1}(0)$, and $a$ and $b$ belong to some finite sets.  
Then, with slight abuse of notation, $\Delta_j$ is now defined such that $\Delta_j[(x,a),(y,b)]=1$ if $x_j\neq y_j$, and $\Delta_j[(x,a),(y,b)]=0$ otherwise. Specifically, in our case rows of $\Gamma$ will be indexed by $(x,M)$ for some $x\in f^{-1}(1)$ and $M\in\cal C$, and columns will simply be indexed by $y\in f^{-1}(0)$.

Second, $\Gamma$ will be the submatrix of a larger matrix $\widetilde{\Gamma}$ (defined below), whose rows are indexed by the elements of $[q]^n\times\mathcal{C}$ and whose columns are indexed by $[q]^n$.
Then $\Delta_j$ is naturally extended to all $x,y\in[q]^n$ and $M\in\mathcal{C}$ by $\tilde{\Delta}_j[(x,M),y]=1$ if $x_j\neq y_j$, and $\tilde{\Delta}_j[(x,M),y]=0$ otherwise. Since $\Gamma\circ\Delta_J$ is a submatrix of $\widetilde{\Gamma}\circ\tilde{\Delta}_J$, we will have $\norm{\Gamma\circ\Delta_J}\leq\norm{\widetilde{\Gamma}\circ\tilde{\Delta}_J}$. Hence we only need to upper bound the latter norm.

We now define $\widetilde{\Gamma}$.
Consider the Hilbert space $\mathbb{C}^q$. Let $E_0$ denote the orthogonal projector onto the vector
$\frac{1}{\sqrt{q}}(1,1,\ldots,1)$, and $E_1=\mathrm{I}-E_0$ its orthogonal complement.
For every $S\subseteq[n]$, let $E_S=\otimes_{j\in[n]}E_{s_j}$
{(acting on $\mathbb{C}^{q^n}$)}, where $s_j=1$ if $j\in S$, and $s_j=0$ otherwise.
Note that $E_SE_{S'}=E_S$ if $S=S'$, and $E_SE_{S'}=0$ otherwise.
Define $\widetilde{\Gamma}$ as
$$
\widetilde{\Gamma}=\left[\begin{array}{c} \vdots \\ G_M \\ \vdots \end{array}\right]_{M\in\mathcal{C}},\quad\text{with }
G_M=\sum_{S\subseteq [n]}\alpha_S(M) E_S,
$$ 
where the $\alpha_S(M)$ come from a feasible solution to the dual~\eqref{eqlgcdual}. We then define $\Gamma$ as the submatrix of $\widetilde{\Gamma}$ obtained by keeping only those columns indexed by $y$ such that $f(y)=0$; and only those rows indexed by $(x,M)$ such that $M$ is a $1$-certificate for $x$ (i.e., for all $z\in [q]^n$, $z_M=x_M$ implies $f(z)=1$). 
{\begin{fact}
$\displaystyle
\norm{\Gamma}\geq\sqrt{\frac{1}{2}\sum_{M\in{\cal C}} \alpha_{\emptyset}(M)^2}.$
\end{fact}

\begin{proof}
Belovs and Rosmanis~\cite[Lemma~17]{belovs:nonadaptive} prove this result for any matrix $\Gamma$ constructed as above assuming that, for each $M\in{\cal C}$,
(1)  $\alpha_S(M)=0$ whenever $M\subseteq S$, and (2) $M$ is equipped with an orthogonal array $T_M$ of length $|M|$. Those two assumptions are satisfied in our case too.
\end{proof}}

Upper bounding $\norm{\widetilde{\Gamma}\circ\tilde{\Delta}_J}$ requires some additional steps compared to~\cite{belovs:nonadaptive}.
We first review the approach of~\cite{belovs:nonadaptive}, which is for the special case $J=\{j\}$. 
Define a linear map $\varphi_j$ on matrix $\widetilde{\Gamma}$ by its action on blocks $E_S$, for every $S\subseteq[n]$.
First, let $\varphi$ be such that $\varphi(E_0)=E_0$ and $\varphi(E_1)=-E_0$.
Then $\varphi_j(E_S)=
E_{s_1}\otimes\ldots \otimes E_{s_{j-1}} \otimes \varphi(E_{s_j}) \otimes E_{s_{j+1}}\otimes\ldots\otimes E_{s_n}$.
An alternative but equivalent definition is
$$\varphi_j(E_S)=\begin{cases}
E_S,&\text{ if $j\not\in S$;}\\
\displaystyle-E_{S\setminus\{j\}}&\text{ otherwise.}
\end{cases}$$
The map $\varphi_j$ was introduced because it satisfies
$E_S\circ\Delta_{j} =\varphi_j(E_S)\circ\Delta_j$. This comes from the observation that
$\varphi(E_1)\circ \Delta_1=E_1\circ \Delta_1$, since $E_1=\mathrm{I}-E_0$ and
$\mathrm{I}\circ \Delta_1=0$.
The approach of~\cite{belovs:nonadaptive} then consists of
applying $\varphi_j$ to $\widetilde{\Gamma}$ before computing the norm of $\widetilde{\Gamma}\circ\tilde{\Delta}_j$. 

We now generalize $\varphi_j$ to subsets $J\subseteq [n]$ as
$$
\varphi_J(E_S)=\begin{cases}
E_S,&\text{ if $J\not\subseteq S$;}\\
\displaystyle-\sum_{S':S\setminus J \subseteq S'\subsetneq S} E_{S'},&\text{ otherwise.}
\end{cases}
$$
Then $\varphi_j$ satisfies the following fact, which is
an extension of the case $J=\{j\}$ (proved in~\cite{belovs:nonadaptive}). 
\begin{fact}\label{fb}
Let $J\subseteq [n]$ be any subset.
Then $ \widetilde{\Gamma}\circ\tilde{\Delta}_J = \varphi_J(\widetilde{\Gamma})\circ\tilde{\Delta}_J$.
\end{fact}

{\begin{proof}
By linearity it suffices to prove the fact for $E_S$, i.e., that $E_S\circ \tilde{\Delta}_J = \varphi_J (E_S) \circ \tilde{\Delta}_J$, where $S,J\subseteq [n]$.
This equality is immediate when $J\not\subseteq S$, since then $\varphi_J (E_S)=E_S$.

Assume from now on that $J\subseteq S$. 
For notational simplicity, assume further and without loss of generality that $J=\{1,2,\ldots,j\}$, and set $F=E_{s_{j+1}}\otimes\ldots\otimes E_{s_n}$, hence $E_S= E_1^{\otimes^j} \otimes F$. Using $\mathrm{I}=E_0+E_1$ we have 
$$ 
\mathrm{I}^{\otimes j} =(E_0+E_1)^{\otimes^j} = \sum_{S':\emptyset \subseteq S'\subseteq J} E_{S'},
$$
where the notations $E_J$ and $E_{S'}$ stand for the first $j$ bits only.
This implies
$$
E_J= \mathrm{I}^{\otimes j}- \sum_{S':\emptyset \subseteq S'\subsetneq J} E_{S'}.
$$
This concludes the proof since 
$$
(\mathrm{I}^{\otimes j} \otimes F) \circ \tilde{\Delta}_J=0,
\quad\text{and}\quad
\varphi_J(E_S) = - \sum_{S':\emptyset \subseteq S'\subsetneq J} E_{S'}\otimes F.
$$
\end{proof}}

Therefore (using also Fact~\ref{optp}) we can upper bound $\norm{\widetilde{\Gamma}\circ\tilde{\Delta}_J}$ by $2\norm{\varphi_J(\widetilde{\Gamma})}$.
It remains to upper bound the latter norm.

{\begin{fact}
$\displaystyle\norm{\varphi_J(\widetilde{\Gamma})} = \max_{S\subseteq [n]} \sqrt{\sum_{M\in{\cal C}} (\alpha_S(M)-\alpha_{S\cup J}(M))^2}.$
\end{fact}}

\begin{proof}
We first compute $\varphi_J(G_M)$:
$$
\varphi_J(G_M)=\sum_{S\subseteq [n]}\beta_S(M) E_S,
\quad\text{where }\beta_S(M)=\alpha_S(M)-\alpha_{S\cup J}(M).
$$
Observe that $\beta_S(M)=0$ if $J\subseteq S$.
Now rewrite $(\varphi_J(\widetilde{\Gamma}))^*\varphi_J(\widetilde{\Gamma})$ as
$$
(\varphi_J(\widetilde{\Gamma}))^*\varphi_J(\widetilde{\Gamma})
=\sum_{M\in\mathcal{C}} (\varphi_J(G_M))^*\varphi_J(G_M)
= \sum_{S\subseteq[n]} \left(\sum_{M\in\mathcal{C}}\beta_S(M)^2 \right) E_S.
$$
Since the different $E_S$ project onto orthogonal subspaces, we can conclude 
$$
\norm{\varphi_J(\widetilde{\Gamma})} = \sqrt{\norm{(\varphi_J(\widetilde{\Gamma}))^*\varphi_J(\widetilde{\Gamma})}}=\max_{S\subseteq[n]} \sqrt{\sum_{M\in\mathcal{C}}\beta_S(M)^2}.
$$
\end{proof}

{We therefore have}
$$
\norm{\Gamma\circ\Delta_J}\leq\norm{\widetilde{\Gamma}\circ\tilde{\Delta}_J}\leq 2\norm{\varphi_J(\widetilde{\Gamma})}= 2\max_{S\subseteq [n]} \sqrt{\sum_{M\in{\cal C}} (\alpha_S(M)-\alpha_{S\cup J}(M))^2}.
$$
When $J$ has size at most~$p$, the right-hand side is at most~$2$ because of the constraint parallel-\eqref{eqlgcdualconstraint}, applied to edge $(S,J')\in \mathcal{E}_p$ with $J'=J\setminus S$.
Therefore
$$
\ADV^{p\parallel}(f)\geq\frac{\norm{\Gamma}}{\max_{J\in{[n]\choose p}}\norm{\Gamma\circ \Delta_J}}
\geq \frac{\sqrt{\frac{1}{2}\sum_{M\in{\cal C}} \alpha_{\emptyset}(M)^2}}{2}=\frac{1}{2\sqrt{2}}\LGC^{p\parallel}({\cal C}).
$$ 

\section{Proof of Theorem \ref{thksumparallelLB}}\label{appksumparallelLB}

Here we prove our parallel lower bound for the $\kSUM$ problem.
The proof strategy is the same as in Theorem~\ref{thEDparallelLB}.
We now use certificate structure ${\cal C}={[n]\choose k}$ with the orthogonal array 
$$
T=\{(v_1,\ldots,v_k) : \sum_{i=1}^k v_i=0 \mod q\}.
$$ 
This induces the $\kSUM$ problem in the way mentioned in Theorem~\ref{thq=p-lgc}.
We define the following solution to the dual for~$\LGC^{p\parallel}({\cal C})$:
\begin{quote}
$\displaystyle\alpha_j=\frac{1}{2n^{k/2}}\max((n/p)^{k/(k+1)} - j/p, 0)$\\[1mm]
$\displaystyle\alpha_S(M)=0$ if $M\subseteq S$\\[1mm]
$\displaystyle\alpha_S(M)=\alpha_{|S|}$ otherwise
\end{quote}
Fix some $e=(S,J)$ with $S\subseteq[n]$ of size~$s$, and disjoint $J\subseteq[n]$ of size at most $p$. 
Let $L$ denote the left-hand side of constraint parallel-\eqref{eqlgcdualconstraint}. In order to establish that the above solution is feasible, we want to show $L\leq 1$.
With respect to~$e$, we can distinguish different kinds of $M=\{i_1,\ldots,i_k\}$, depending on $i:=|M\cap S|$ and $j:=|M\cap J|$:
\begin{enumerate}
\item $i+j<k$. There are ${s\choose i}{|J|\choose j}$ such $M$, and each contributes $\leq|\alpha_s-\alpha_{s+|J|}|^2\leq 1/4n^k$ to~$L$.
\item $i+j=k$. There are ${s\choose i}{|J|\choose j}$ such $M$, each contributes $\alpha_s^2$ to~$L$ if $i<k$, and~0 if $i=k$.\\ 
Over all such choices of $i$ and $j$, at most $|J|\binom{s+|J|-1}{k-1}$ of these $M$ have $j\geq 1$ (i.e., $\alpha_S(M)\neq 0$), since this counts the number of ways of choosing one index from $J$, and $k-1$ more from $J\cup S$. 
\end{enumerate}
Note that $\alpha_s=0$ if $s\geq p(n/p)^{k/(k+1)}$, so below we may assume $s+p-1\leq 2p(n/p)^{k/(k+1)}$ for $n$ sufficiently large.
Also $\alpha_s\leq\alpha_0=(n/p)^{k/(k+1)}/2n^{k/2}$.
Hence we can bound~$L$ as
\begin{align*}
L & \leq  \sum_{i=0}^{k-1} \sum_{j=0}^{k-1-i} {s\choose i}{|J|\choose j}|\alpha_s-\alpha_{s+|J|}|^2 + \sum_{i=0}^{k-1} {s\choose i}{|J|\choose k-i}\alpha_s^2\\
& \leq  \sum_{\ell=0}^{k-1}{s+|J|\choose \ell}|\alpha_s-\alpha_{s+|J|}|^2 +  |J|{s+|J| - 1\choose k-1}\alpha_s^2\\
& \leq  \frac{n^{k-1}}{4n^k} +  \frac{p(s+p-1)^{k-1}}{(k-1)!}\frac{(n/p)^{2k/(k+1)}}{4n^k}\\
& \leq  \frac{1}{4n}+\frac{2^{k-1}}{(k-1)!}p^k(n/p)^{(k-1)k/(k+1)}\frac{(n/p)^{2k/(k+1)}}{4n^k}\\
& \leq  \frac{1}{4n}+p^k\frac{(n/p)^k}{4n^k} = \frac{1}{4n}+\frac{1}{4}\leq 1.
\end{align*}
Hence our solution is feasible. Its objective value is 
$\displaystyle\sqrt{{n \choose k} \alpha_0^2}=\Omega\left((n/p)^{k/(k+1)}\right)$.

\section{Parallel quantum oracle interrogation}\label{app:parallelinter}

The following quantum algorithm recovers the complete input $x\in\01^n$ with high probability, using roughly $n/2p$ $p$-parallel queries:
\begin{enumerate}
\item With $T=n/2+O(\sqrt{n\log(1/\eps)})$ and $B=\sum_{i=0}^{T}{n\choose i}$ being the number of $y\in\01^n$ with weight $|y|\leq T$,
set up the $n$-qubit superposition
$\frac{1}{\sqrt{B}}\sum_{y\in\01^n:|y|\leq T}\ket{y}.$
\item Apply the unitary $\ket{y}\mapsto(-1)^{x\cdot y}\ket{y}$. We can implement this using $\ceil{T/p}$ $p$-parallel queries for $|y|\leq T$: the first batch of $p$ queries would query the first $p$ positions where $y$ has a one and put the answer in the phase; the second batch queries the next $p$ positions, etc.
\item Apply a Hadamard transform to all qubits and measure.
\end{enumerate}
To see the correctness of this algorithm, note that the fraction of $n$-bit strings $y$ that have weight $>T$ is $\ll\eps$.
Hence the state obtained in step~2 is very close to the state $\frac{1}{\sqrt{2^n}}\sum_{y\in\01^n}(-1)^{x\cdot y}\ket{y}$,
whose Hadamard transform is exactly~$\ket{x}$.

\end{document}